\documentclass[a4paper,UKenglish]{lipics-v2016}
\usepackage{microtype}

\usepackage[ruled,linesnumbered, noend]{algorithm2e}
\usepackage{amsmath}

\newtheorem{proposition}{Proposition}

\usepackage{xspace}
\newcommand{\absets}{\texttt{ABSE-TS}\xspace}

\usepackage{floatrow}
\usepackage{enumitem}

\newcommand{\apx}{\textsf{APX}}
\newcommand{\np}{\textsf{NP}}
\newcommand{\p}{\textsf{P}}

\newcommand{\F}{\mathcal{F}}
\renewcommand{\H}{\mathcal{H}}
\newcommand{\T}{\mathcal{T}}
\newcommand{\U}{\mathcal{Q}}

\allowdisplaybreaks
\DeclareMathOperator{\UE}{UE}

\title{An Improved Algorithm for Computing All the Best Swap Edges of a Tree Spanner}

\author[1]{Davide Bil\`o}
\author[2]{Feliciano Colella}
\author[3]{Luciano Gual\`a}
\author[4]{Stefano Leucci}
\author[5,6]{Guido Proietti}

\affil[1]{Universit\`a  di Sassari, Italy. \texttt{davide.bilo@uniss.it}}
\affil[2]{Gran Sasso Science Institute, L'Aquila,  Italy. \texttt{feliciano.colella@gssi.it}}
\affil[3]{Universit\`a  di Roma ``Tor Vergata'', Italy. \texttt{guala@mat.uniroma2.it}}
\affil[4]{ETH Z\"urich, Switzerland. \texttt{stefano.leucci@inf.ethz.ch}}
\affil[5]{Universit\`a  degli Studi dell'Aquila, Italy. \texttt{guido.proietti@univaq.it}}
\affil[6]{Istituto di Analisi dei Sistemi ed Informatica, CNR, Roma, Italy.}

\authorrunning{D. Bil\`o et al.}

\Copyright{Davide Bil\`o and Feliciano Colella and Luciano Gual\`a  and Stefano Leucci and Guido Proietti}%mandatory, please use full first names. LIPIcs license is "CC-BY";  http://creativecommons.org/licenses/by/3.0/

\subjclass{G.2.2 [Graph Theory] Graph algorithms, Trees}% mandatory: Please choose ACM 1998 classifications from http://www.acm.org/about/class/ccs98-html . E.g., cite as "F.1.1 Models of Computation".
\keywords{Transient edge failure, Swap algorithm, Tree spanner}% mandatory: Please provide 1-5 keywords

\makeatletter
\let\@DOIPrefix\@empty
\makeatother
\renewcommand{\copyrightline}{\relax}

\begin{document}

\maketitle

\begin{abstract}
A \emph{tree $\sigma$-spanner} of a positively real-weighted $n$-vertex and $m$-edge undirected graph $G$ is a spanning tree $T$ of $G$ which approximately preserves (i.e., up to a multiplicative \emph{stretch factor} $\sigma$) distances in $G$.
Tree spanners with provably good stretch factors find applications in communication networks, distributed systems, and network design. However, finding an optimal or even a good tree spanner is a very hard computational task.  Thus, if one has to face a \emph{transient} edge failure in $T$, the overall effort that has to be afforded to rebuild a new tree spanner (i.e., computational costs, set-up of new links, updating of the routing tables, etc.) can be rather prohibitive. To circumvent this drawback, an effective alternative is that of associating with each  tree edge a best possible (in terms of resulting stretch) \emph{swap edge} -- a well-established approach in the literature for several other tree topologies. Correspondingly, the problem of computing \emph{all} the best swap edges of a tree spanner is a challenging algorithmic problem, since solving it efficiently means to exploit the structure of shortest paths not only in $G$, but also in all the scenarios in which an edge of $T$ has failed. For this problem we provide a very efficient solution, running in $O(n^2 \log^4 n)$ time, which drastically improves (almost by a quadratic factor in $n$ in dense graphs!) on the previous known best result.
\end{abstract}

\section{Introduction}
The problem of computing \emph{all the best swap edges} (ABSE) of a tree has a long and rich algorithmic tradition.
Basically, let $G = (V(G), E(G),w)$ be an $n$-vertex and $m$-edge $2$-edge-connected undirected graph, with edge-weight function $w : E(G) \rightarrow \mathbb{R}^+$, and assume we are given a spanning tree $T$ of $G$, which was computed by addressing some criterion (i.e., objective function) $\phi$. Then, the problem is that of computing a BSE for every edge $e \in E(T)$, namely an edge $f \in E(G) \setminus E(T)$ such that the swap tree $T_{e/f}$ obtained by swapping $e$ with $f$ in $T$ optimizes some objective function $\phi^\prime$ out of all possible swap trees. Quite reasonably, the function $\phi^\prime$ must be related (if not coinciding at all) with $\phi$.

The first immediate motivation for studying an ABSE problem comes from the edge fault-tolerance setting -- a commonly accepted framework. Broadly speaking, the algorithmic question here is to design \emph{sparse} subgraphs that guarantee a proper level of functionality even in the presence of an edge failure. In such a context, the rationale of an ABSE-based solution is the following: operations are normally performed on a (possibly optimal) spanning tree, and whenever an edge failure takes place, a corresponding BSE is plugged in. This way, the connectivity is reestablished in the most prompt and effective possible way (see also~\cite{IIOY03,Pro00} for some additional practical motivations).

Besides their practical relevance, ABSE problems have also an interesting theoretical motivation. Indeed, swapping can be reviewed as an exploration of the space of the perturbed (w.r.t. an edge swap) solutions to a given spanning tree optimization problem. Thus, the algorithmic challenge of solving efficiently an ABSE problem is related with the understanding of the structure of this space of perturbed solutions.
And this is exactly why each ABSE problem has its own combinatorial richness, and thus requires a specific approach to be solved efficiently. Then, different ABSE problems have required the use of completely different approaches and methods in order to obtain efficient solutions. For instance, the most famous and studied ABSE problem comes when $T$ is a \emph{minimum spanning tree} (MST) of $G$. In this case, a best swap is of course a swap edge minimizing the \emph{cost} (i.e., sum of the edge weights) of the swap tree, i.e., a swap edge of minimum weight (and we know this produces a MST of the perturbed graph). This problem is also known as the MST \emph{sensitivity analysis} problem, and can be solved in $O(m\log\alpha(m,n))$ time~\cite{Pet05}, where $\alpha$ denotes the inverse of the Ackermann function, by using an efficient data structure, namely the \emph{split-findmin} \cite{Gabow85}. This was improving on another efficient solution given by Tarjan \cite{Tarjan82}, running in $O(m \, \alpha(m,n))$ time and making use of the \emph{transmuter}, namely a compact way of representing the cycles of a graph. Other data structures which revealed their usefulness to solve efficiently ABSE problems include \emph{kinetic heaps} \cite{BGP15}, \emph{top trees} \cite{BiloCG0P15}, \emph{mergeable heaps} \cite{NPW04}, and many others.

In this paper, we focus on the ABSE problem on the elusive spanning tree structure, namely the \emph{tree spanner} (\absets problem in the following). A tree spanner is built with the aim of preserving node-to-node distances in $G$. Indeed, the \emph{stretch factor} $\sigma$ of a spanning tree $T$ of $G$ is defined as the \emph{maximum}, over all the pairs $u,v \in V(G)$, of $d_T(u,v)/d_G(u,v)$, where $d_T$ and $d_G$ denote distances in $T$ and $G$, respectively. Correspondingly, an \emph{optimal} tree spanner has minimum stretch out of all the spanning trees of $G$.
 Unfortunately, finding an optimal tree spanner is notoriously an \apx-hard problem, with no known $o(n)$-approximation.
 Hence, once a given solution undergoes a transient edge failure, the recomputation from scratch of a new (near) optimal solution is computationally unfeasible.
 Thus, swapping in a tree spanner is even more attractive than in general, and indeed the \absets problem was studied in \cite{DGW10}, where the authors devised two solutions for both the weighted and the unweighted case, running in $O(m^2 \log n)$ and $O(n^3)$ time, respectively, and using $O(m)$ and $O(n^2)$ space, respectively.
 However, there the authors assume that a BSE is an edge minimizing the stretch of the swap tree w.r.t.\ distances in the \emph{original} graph $G$, and not in the graph $G$ deprived of $e$, say $G-e$. This contrasts with the general assumption (and the intuition) that the quality of a swap tree should be evaluated in the surviving graph. 
 Hence, in \cite{BiloCG0P15} the authors resorted to such a standard setting, and provided two efficient linear-space solutions for both the weighted and the unweighted case, running in $O(m^2 \log \alpha(m,n))$ and $O(m n  \log n)$ time, respectively, and both using linear space. Notice that from a computational point of view, as shown in \cite{BiloCG0P15}, the two settings are substantially equivalent, so our solutions can be used to improve the results given in \cite{DGW10} as well.

\subsection{Our result}

In this paper, we present a new algorithm that solves the \absets problem in $O(n^2 \log^4 n)$ time and $O(n^2 + m \log^2 n)$ space. Thus, our solution improves on the running time of both the algorithms provided in \cite{BiloCG0P15}, for weighted and unweighted graphs, respectively, whenever $m = \Omega(n \log^3 n)$. Most remarkably, for dense weighted graphs, the improvement is almost quadratic in $n$.

To put into focus our result, it is worth noticing that, as observed in \cite{DGW10}, the estimation of the stretch of the swap tree induced by a \emph{single} swap edge $f$ for a given failing edge $e$, would in principle ask for the evaluation of the stretch of $O(m)$ relevant pairs of nodes in $G$, namely the endvertices of all the non-tree edges that may serve as swap edge for $e$ besides $f$. And in fact, a \emph{critical edge} for $f$ is the one whose endvertices maximize such a stretch out of these non-tree edges, and two swap edges will be essentially compared on the basis of their stretch w.r.t. their critical edge.
This is basically the reason why both previous approaches take $\Omega(m^2)$ time.
Thus, to avoid such a bottleneck, we drastically reduce, on the one hand, the number of candidate best swap edges, and on the other hand, the number of potential critical edges that need to be checked. More precisely, for each of the $n-1$ considered edges in $T$, we succeed in reducing to $O(n \log n)$ the number of best swap edge candidates, and for each one of them we just need to check $O(\log ^2 n)$ possible critical edges.
The key ingredients to reach such a goal are the following:

\begin{itemize}
\item A \emph{centroid decomposition} of $T$, which consists of a log-depth hierarchical decomposition of the vertices in $T$; a careful use of such a decomposition, combined with a set of preprocessing steps that associate various information with the tree nodes, allows us to reduce the number of candidate BSEs and of their corresponding candidate critical edges. As far as we know, this is the first time that such a decomposition is used to solve an ABSE problem, and we believe it will possibly be useful in other contexts as well.
\item The second ingredient is given by the dynamic maintenance of the \emph{upper envelopes} of a set of linear functions. Each of these functions is associated with a non-tree edge, and whenever the failure of a given tree edge is considered, it expresses the stretch such a non-tree edge  induces w.r.t. a variable candidate BSE. This way, when we have to find a critical edge for a given candidate BSE $f$, we have to select the \emph{maximum} out of all the functions once they are evaluated in $f$. In geometric terms, this translates into the maintenance of the upper envelope of a set of functions, with the additional complication that, for consistency reasons, this set of functions must be suitably partitioned into groups according to the underlying centroid decomposition, and moreover these groups are dynamic, since they depend on the currently considered tree edge.
\end{itemize}

\subsection{Related work}
The research on tree spanners is very active, also due to the strong relationship with the huge literature on \emph{spanners}, where distances in $G$ are approximately preserved through a \emph{sparse} spanning subgraph. As mentioned before, finding an optimal tree spanner is a quite hard problem. More precisely, on weighted graphs, if $G$ does not admit a tree $1$-spanner (i.e., a spanning tree with $\sigma=1$, which can be established in polynomial time \cite{CC95}), then the problem is not approximable within any constant factor better than $2$, unless \p=\np\ \cite{LW08}.
In terms of approximability, no non-trivial upper bounds are known, except for the $O(n)$-approximation factor returned by a \emph{minimum spanning tree} (MST) of $G$.
If $G$ is \emph{unweighted}, things go slightly better. More precisely, in this case the problem becomes $O(\log n)$-approximable, while unless \p=\np, the problem is not approximable within an additive term of $o(n)$ \cite{EP08}. Moreover, the corresponding decision problem of establishing whether $G$ admits a tree spanner with stretch $\sigma$ is \np-complete for every fixed $\sigma \geq 4$ (for $\sigma=2$ it is polynomial-time solvable \cite{CC95}, while for $\sigma=3$ the problem is open).
Finally, it is known that constant-stretch tree spanners can be found for several special classes of (unweighted) graphs, like strongly chordal, interval, and permutation graphs (see \cite{BCD99} and the references therein).

Concerning the problem of swapping in spanning trees, this has received a significant attention from the algorithmic community. There is indeed a line of papers that address ABSE problems starting from different types of spanning trees. Just to mention a few, besides the MST, we recall the \emph{minimum diameter spanning tree} (MDST), the \emph{minimum routing-cost spanning tree} (MRCST), and the \emph{single-source shortest-path tree} (SPT).
Concerning the MDST, a best swap is instead an edge minimizing the \emph{diameter} of the swap tree~\cite{IR98,NPW01}, and the best solution runs in $O(m \log \alpha(m,n))$ time \cite{BGP15}. Regarding the MRCST, a best swap is clearly an edge minimizing the \emph{all-to-all routing cost} of the swap tree~\cite{WHC08}, and the fastest solution for solving this problem has a running time of $O\left(m 2^{O(\alpha(n,n))}\log^2 n\right)$~\cite{BGP14}. Concerning the SPT, the most prominent swap criteria are those aiming to minimize either the maximum or the average distance from the root, and the corresponding ABSE problems can be addressed in  $O(m \log \alpha(m,n))$ time \cite{BGP15} and $O(m \, \alpha(n,n) \log^2 n)$ time \cite{DP07}, respectively. Recently, in \cite{BCGLPsirocco17}, the authors proposed two new criteria for swapping in a SPT, which are in a sense related with this paper, namely the minimization of the maximum and the average stretch factor from the root, for which they proposed an efficient $O(m n  +n^2 \log n)$ and $O(m n \log \alpha(m,n))$ time solution, respectively.

Finally, for the sake of completeness, we mention that for the related concept of \emph{average} tree $\sigma$-spanners, where the focus is on the average stretch w.r.t. all node-to-node distances, it was shown that every graph admits an average tree $O(1)$-spanner \cite{ABN07}.

\subsection{Preliminary definitions}
Let $G = ( V(G), E(G), w)$ be a $2$-edge-connected, edge-weighted, and undirected graph with cost function $w : E(G) \rightarrow \mathbb{R}^+$. We denote by $n$ and $m$ the number of vertices and edges of $G$, respectively. If $X \subseteq V(G)$, let $E(X)$ be the set of edges incident to at least one vertex in $X$. When $X=\{ v \}$, we may write $E(v)$ instead of $E(\{ v \})$. Given an edge $e \in E(G)$, we will denote by $G-e$ the graph obtained from $G$ by removing edge $e$. Similarly, given a vertex $v \in V(G)$, we will denote by $G-v$ the graph obtained from $G$ by removing vertex $v$ and all its incident edges.
Given an edge $e \in E(T)$, we let $S(e)$ be the set of all the \emph{swap edges} for $e$, i.e., all edges in $E(G) \setminus \{ e \}$ whose endpoints lie in two different connected components of $T-e$. We also define $S(e, X)=S(e) \cap E(X)$, and $S(e, X, Y)=S(e) \cap E(X) \cap E(Y)$. When $X=\{ v \}$, we will simply write $S(e,v)$ in lieu of $S(e,\{v\})$.
For any $e \in E(T)$ and $f \in S(e)$, let $T_{e/f}$ denote the \textit{swap tree} obtained from $T$ by replacing $e$ with $f$.

Given two vertices $x,y \in V(G)$, we denote by $d_G(x,y)$ the \emph{distance} between $x$ and $y$ in $G$. We define the \textit{stretch factor} of the pair $(x,y)$ w.r.t. $G$ and $T$ as $\sigma_G(T, x, y) = \frac{d_T(x,y)}{d_G(x,y)}$. Accordingly, the stretch factor $\sigma_G(T)$ of $T$ w.r.t. $G$ is defined as $\sigma_G(T) = \max_{x,y \in V(G)} \sigma_G(T, x, y)$.

\begin{definition}[Best Swap Edge]
	An edge $f^* \in S(e)$ is a \textit{best swap edge} (BSE) for $e$ if $f^* \in \arg\min_{f \in S(e)} \sigma_{G-e}(T_{e/f})$.
\end{definition}

In the sequel, in order to solve the \absets problem, we will show how to efficiently find a BSE for every edge $e$ of a tree spanner $T$ of $G$.

\section{High-level description of the algorithm}

It is useful to consider the tree $T$ as rooted at any fixed vertex, and to assume, w.l.o.g., that $T$ is binary. Indeed, if $T$ is not binary, then it is possible, by using standard techniques, to transform $G$ and $T$ into an equivalent graph $G^\prime$ and a corresponding binary spanning tree $T^\prime$, with $|V(G^\prime)| = \Theta(n)$ and $|E(G^\prime)| = \Theta(m)$, and such that a BSE for any edge of $T$ is univocally associated with a BSE for a corresponding edge of $T'$. This transformation requires linear time and it is sketched in Appendix~\ref{sec:degree_reduction} for the sake of completeness.

As a preprocessing step, we compute a \emph{centroid decomposition} of $T$.
A \emph{centroid} of an $n$-vertex tree is a vertex whose removal splits $T$ into subtrees of size at most $n/2$ \cite{jordan1869assemblages}.
A centroid decomposition of $T$ can be computed in $O(n \log n)$ time, and can be represented by a tree $\T$ of height $O(\log n)$, whose nodes are actually subtrees of $T$.
$\T$ is recursively defined as follows: the root of $\T$ is $T$. Then, let $\tau$ be a node of $\T$ (i.e., a subtree of $T$) such that $\tau$ contains more than one vertex, and let $c$ be a centroid of $\tau$. Since $T$ is binary, the forest $\tau-c$ contains at most $3$ trees, that we call $\tau_c^1$, $\tau_c^2$, and $\tau_c^3$ (if $\tau-c$ generates less than $3$ subtrees, we allow some $\tau_c^i$ to be the empty tree).
Moreover, let $\tau_c^0$ be the subtree of $T$ containing the sole vertex $c$. Then, $\tau$ will have in $\T$ a child for each of the subtrees $\tau_c^i, i=0,\ldots,3$ (see Figure~\ref{fig:bce} (a)).
Since a centroid on a $n$-vertex tree can be found in linear time, the whole procedure requires $O(n \log n)$ time, and it is easy to see that the height of $\T$ is $O(\log n)$.

Our solution (see Algorithm~\absets) works in $n-1$ phases, one for each tree edge as considered in preorder w.r.t. $T$, and at the end of each phase returns a BSE for that edge. 
Let $e \in E(T)$ be the currently considered edge, and let $U_e$ (resp. $D_e$) be the set of vertices that belong to the connected component of $T-e$ that contains (resp. does not contain) the root of $T$. We break down each of these phases into $O(n)$ additional sub-phases: when edge $e$ is failing, we consider all the vertices in $U_e$ and, for each such vertex $v$, we solve a restricted version of the \absets problem
where we compute: (i) a \emph{$v$-restricted best swap edge} ($v$-BSE for short), i.e., an edge $f \in \arg\min_{f \in S(e,v)} \sigma_{G-e}(T_{e/f})$, and (ii) the corresponding stretch factor $\sigma_{G-e}(T_{e/f})$.  To simplify handling of special cases, whenever $S(e, v) = \emptyset$, we assume that $f=\bot$ and that $\sigma_{G-e}(T_{e/f}) = +\infty$. As we will see in the rest of the paper, the core of our algorithm is exactly the efficient computation of these $v$-BSEs  and of their stretch factors. This is done trough a clever selection of a small set of \emph{candidate} $v$-BSEs, as we will discuss in more detail in the next section.
Once all the $v$-BSEs for $e$ are computed, a BSE for $e$ can be found as the one minimizing the associated stretch factor.

\begin{algorithm}[t]{\caption{ABSE-TS($G$, $T$)}}
	\label{alg:absets}
\footnotesize{	$\T \gets$ Centroid decomposition of $T$\;
	\ForEach(\tcp*[f]{$n-1$ phases}){$e \in E(T)$ in postorder}
	{
		$U_e \gets $ vertices of the component of $T-e$ that contains the root of $T$\;
		
		$f^* \gets \bot$\tcp*{Current BSE for $e$}
		\ForEach(\tcp*[f]{$O(n)$ sub-phases}){$v \in U_e$}
		{
			\textbf{compute} a $v$-BSE $f$ for $e$ and the corresponding stretch factor\tcp*{This takes $O(\log^4 n)$ time by using $\T$ and the dynamic maintenance of the upper envelopes associated with the swap edges, as shown in Section \ref{sec:BSE}}
			
			\lIf{ $\sigma_{G-e}(T_{e/f}) < \sigma_{G-e}(T_{e/f^*})$ }
			{$f^* \gets f$}
		}
		\Return {\textup{$f^*$ as BSE for $e$ and \textbf{continue} with the next phase.}}

	}}
\end{algorithm}

\section{Computing efficiently a \texorpdfstring{$v$-BSE}{v-BSE}}
\label{sec:BSE}
To show how a $v$-BSE for $e$ can be computed efficiently, we need some preliminary definitions:

\begin{definition}[Critical Edge]
Given $e \in E(T)$ and a swap edge $f = (v, u) \in S(e, v)$, a \emph{critical edge}\footnote{Notice that this definition does not contain $d_{G-e}(x,y)$ at the denominator, as expected, since it already incorporates the property stated in the forthcoming Proposition \ref{prop:min_bse}.} for $f$ is an edge $g = (x,y) \in S(e)$ maximizing
$ \displaystyle
	\phi(f, g) := \frac{d_T(x, v) + w(f) + d_T(u, y)}{w(g)}.
$
\end{definition}

\begin{definition}[Best Cut Edge]
	A $v$-\emph{best cut edge} for $e$ ($v$-BCE) is an edge $f \in S(e,v)$ minimizing $\varphi_e(f)=\max_{g \in S(e)} \phi(f, g)$.
\end{definition}

Then, we will make use of the following property, which was given in \cite{BiloCG0P15}:
\begin{proposition}
\label{prop:min_bse}
 Every $v$-BCE for $e$ is a $v$-BSE for $e$.
\end{proposition}

Let us first provide a high-level description of how we compute a $v$-BCE (i.e., a $v$-BSE) for $e$. The algorithm will compute $O(\log n)$ $v$-BCE \emph{candidates}, the best of which will be a $v$-BCE for $e$. Informally speaking, each candidate $f$ will be a swap edge close to the centroid of a certain subtree $\Lambda$ of $T$. Depending on the position of a critical edge for $f$, the algorithm will recurse on a subtree of $\Lambda$ and it will look for the next candidate. Thanks to the centroid decomposition of $T$, the number of recursions/candidates will then be $O(\log n)$.

\begin{figure}[t]
	\centering
	\includegraphics[scale=0.73]{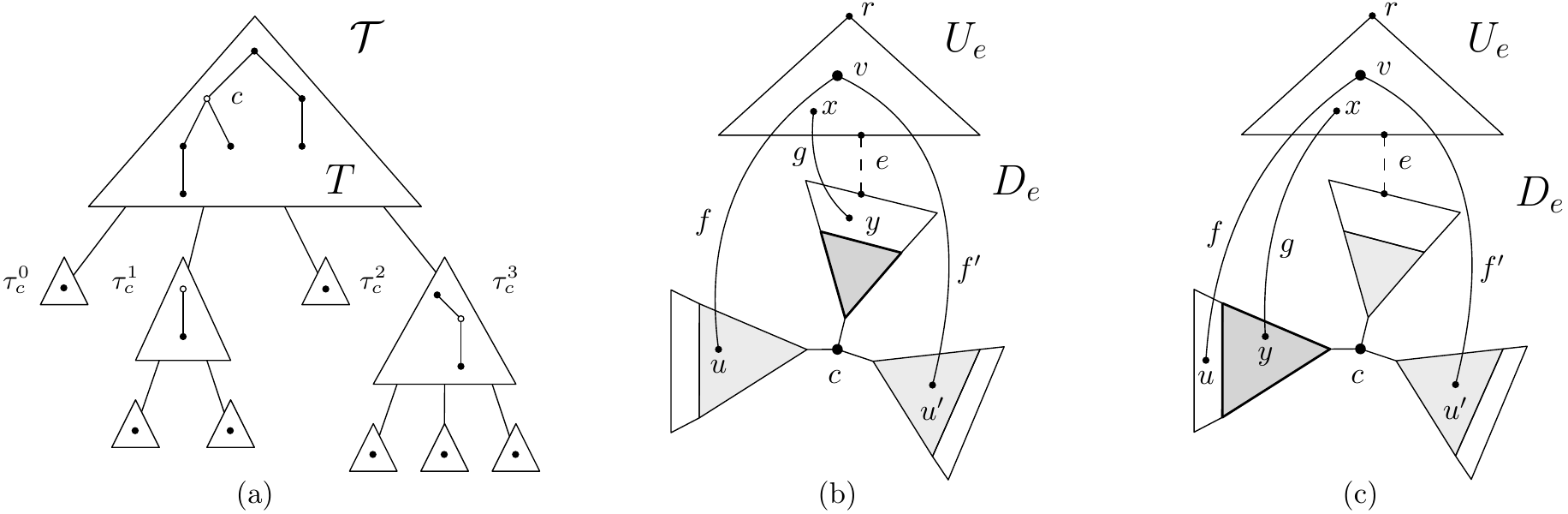}
	\caption{(a) An example of centroid decomposition of the tree $T$ (which corresponds to the first vertex of $\T$).
	(b) and (c): Two of the four possible cases situation illustrated in Lemma~\ref{lemma:BCE_position}.
	The subtree $\widehat{T}$ is represented by the three gray triangles along with the vertex $c$. $f$ is a swap edge for $e$ that minimizes $w(f) + d_T(u, c)$, and $g$ is its corresponding critical edge. The $(c,y)$-\emph{tree} of $\widehat{T}$ is drawn in bold. Notice that $f$ and $g$ do not need to be incident to $\widehat{T}$.}
	\label{fig:bce}
\end{figure}

The key ingredient for the correctness of our algorithm is the next lemma.
Given a subtree  $\widehat{T}$ of $T$, a vertex $c \in V(\widehat{T})$, and a vertex $y \in V(T)$,
	consider the first vertex $z$ of the unique path from $y$ to $c$ in $T$ that also belongs to $V(\widehat{T})$. 	
	The $(c,y)-$\emph{tree} of $\widehat{T}$ is defined as follows: (1) if $z=c$, then it is the empty tree; otherwise (2) it is the tree of the forest $\widehat{T}-c$ that contains $z$. Then, the following holds (see also Figure \ref{fig:bce}~(b)~and~(c)):
	
\begin{lemma}
	\label{lemma:BCE_position}
Let $\widehat{T}$ be a subtree of $T$ such that $V(\widehat{T}) \subseteq D_e$,
and let $c \in V(\widehat{T})$. Moreover, let $f =(v,u) \in S(e, v)$ be a swap edge for $e$ that minimizes $w(f) + d_T(u, c)$, and let $g=(x,y)$ be a critical edge for $f$.
Assume that $S(e,v, V(\widehat{T}) )$ contains a $v$-BCE for $e$.
If $f$ is not a $v$-BCE for $e$, then $S(e,v, V(T') )$ contains a $v$-BCE for $e$, where $T'$ is the
$(c,y)$-tree of $\widehat{T}$.
\end{lemma}
\begin{proof}
	Suppose that $f$ is not a $v$-BCE for $e$, we will show that no swap edge $f^\prime = (v, u^\prime) \in S(e,v)$ with $u^\prime \not\in V(T')$ can be a $v$-BCE for $e$.
	Indeed:		
	\begin{align*}
		\hspace{-2.7pt} \varphi(f^\prime)
		 & \ge \phi(f^\prime, g)
		 = \frac{d_T(x, v) + w(f^\prime) + d_T(u^\prime, y)}{w(g)}
		 = \frac{d_T(x, v) + w(f^\prime) + d_T(u^\prime, c) + d_T(c, y)}{w(g)} \\
		 & \ge \frac{d_T(x, v) + w(f) + d_T(u, c) + d_T(c, y)}{w(g)}
		 \ge \phi(f, g) = \varphi(f),
	\end{align*}
	\noindent where we used the fact that $d_T(u^\prime, y) = d_T(u^\prime, c) + d_T(c, y)$ as either $u^\prime=c$ or $u^\prime$ and $y$ are in two different connected components of $T-c$.
\end{proof}

Lemma~\ref{lemma:BCE_position} allows us to design a recursive algorithm for computing a $v$-BCE for $e$, whose key steps are highlighted in Procedure~\texttt{\ref{alg:find_BCE}} (notice that $v$ and $e$ are fixed). More precisely, the algorithm takes a tree $\Lambda$ of the centroid decomposition $\T$ such that $V(\Lambda) \cap D_e \neq \emptyset$, and it computes a pair $(f^*,g^*)$ such that if $S(e, v, V(\Lambda) \cap D_e)$ contains a $v$-BCE for $e$, then $f^*$ is a $v$-BCE for $e$, and $g^*$ is its critical edge. Procedure~\texttt{\ref{alg:find_BCE}} makes use of an additional function \texttt{\ref{alg:find_critical}}($f, T$) that returns a critical edge for $f$ w.r.t. the failure of $e$. The initial call will be \texttt{FindBCE}($T$). In order to handle base cases, we assume $\phi(\bot,\bot)=+ \infty$.

\begin{procedure}[t]{\caption{FindBCE($\Lambda$)}}
	\label{alg:find_BCE}
	\footnotesize{	\lIf{$|V(\Lambda)| = 0$}{ \Return $(\bot, \bot)$ }

	$c \gets $ Centroid of $\Lambda$\;

	\If{$c \in U_e$}
	{
		$\tau \gets$  unique child of $\Lambda$ in $\T$ that contains all the vertices in $V(\Lambda) \cap D_e$\label{ln:tau}\;
		\Return \FindBCE($\tau$)\;
	}	
	\Else(\tcp*[f]{$c \in D_e$})
	{
		Compute an edge $ f = (v, u) \in \arg \min_{(v,u) \in S(e,v)} \{ w(v,u) + d_T(u, c) \}$\label{ln:find_f}\tcp*{See Appendix~\ref{sec:step 7}}
		$g_1 = (x,y) \gets $ \texttt{FindCritical}($f, T$)\label{ln:find_critical}\tcp*{Compute a critical edge for $f$ (see Sec.~\ref{sec:critical})}
	
		$\tau \gets (c,y)$-tree of $\Lambda$\tcp*{Either $\tau$ is empty or it is a child of $\Lambda$ in $\T$}		
		$(f^\prime, g_2) \gets $ \FindBCE($ \tau$)\;

		\leIf{$\phi(f, g_1 ) \le \phi(f', g_2 )$}{\Return $(f, g_1)$;}{\Return $(f^\prime, g_2)$}
	}}
\end{procedure}

We now prove the correctness of the procedure:
\begin{lemma}
	Procedure~\texttt{\ref{alg:find_BCE}}$(T)$ computes a $v$-BCE for $e$.
\end{lemma}
\begin{proof}
	Consider an invocation of the procedure and let $\Lambda$ and $(f^*, g^*)$ be its parameter and the edges it returns, respectively.
	We prove the following claim by induction on the cardinality of $V(\Lambda)$: if $S(e,v,V(\Lambda) \cap D_e)$ contains a $v$-BCE for $e$, then $f^*$ is a $v$-BCE for $e$ and $g^*$ is a critical edge for $f^*$.

	If $|V(\Lambda)|=0$, then the claim trivially holds.
	Otherwise, $|V(\Lambda)| > 0$, and we distinguish two cases depending on the position of the centroid $c$ of $\Lambda$.
	If $c \in U_e$, then there is only one child $\tau_c^j$ of $\Lambda$ in $\T$ that contains all the vertices in $V(\Lambda) \cap D_e$, as otherwise the vertices in $D_e$ would be disconnected in $\Lambda$. Hence, if $S(e,v,V(\Lambda) \cap D_e)$ contains a $v$-BCE for $e$, then $S(e,v,V(\tau_c^j) \cap D_e)$ also contains a $v$-BCE for $e$, and the claim follows by the inductive hypothesis (as $|V(\tau_c^j)| < |V(\Lambda)|$).
	The remaining case is the one in which $c \in D_e$, here the claim follows from Lemma~\ref{lemma:BCE_position} (where now $\widehat{T}$ is the subtree of $T$ induced by $V(\Lambda) \cap D_e$) together with the inductive hypothesis.
\end{proof}

Next lemma provides an upper bound to the running time of the procedure:
\begin{lemma}\label{lemma:FindBCE}
	Procedure~\texttt{\ref{alg:find_BCE}}$(T)$ requires $O( (\Gamma_f+\Gamma_{\emph{\texttt{FC}}}) \log n)$ time, where $\Gamma_f$ and $\Gamma_{\emph{\texttt{FC}}}$ is the time required to perform Steps \ref{ln:find_f} and \ref{ln:find_critical}, i.e., the time to find edge $f$, and to execute Procedure~\texttt{\ref{alg:find_critical}}, respectively.
\end{lemma}
\begin{proof}
First of all, notice that Step \ref{ln:tau} can be performed in $O(1)$ time, after a $O(\log n)$ preprocessing time in which we mark all the nodes of $\T$ on the path between the leaf of $\T$ containing the lower vertex of $e$ (which clearly belongs to $D_e$) and the root of $\T$. Then, we only need to bound the depth of the recursion of the call \texttt{\ref{alg:find_BCE}}$(T)$.
Observe that each time Procedure~\texttt{\ref{alg:find_BCE}}($\Lambda$) recursively invokes itself on a tree $\Lambda'$, we have that $\Lambda'$ is a child of $\Lambda$ in $\T$.
The claim follows since the height of $\T$ is $O(\log n)$.
\end{proof}

Actually, the time to execute Step \ref{ln:find_f} is $O(\log n)$, after a preprocessing time and space of $O(n^2)$, by making use of \emph{top-trees} \cite{AHLT05}. Due to space limitations, this result is postponed to the appendix. On the other hand, Procedure~\texttt{\ref{alg:find_critical}} will require $O(\log ^3 n)$ time and $O(m \log ^2 n)$ space, as we will show in the next two subsections.

\subsection{Computing a critical edge for \texorpdfstring{$f$}{f}}
\label{sec:critical}
We will compute $O(\log^2 n)$ critical edge candidates for $f$ and we will show that a critical edge for $f$ will be one of them. More precisely, we look at $O(\log n)$ subtrees of the centroid decomposition $\T$ and, for each such subtree $\Lambda$, we will consider $O(\log n)$ subtrees $\Psi$ to find a critical edge candidate having one endpoint in $\Psi$ and the other in $\Lambda$. The choice of the $O(\log^2 n)$ pairs of trees is guided by the position of $f$, while the computation of a candidate for a given pair $(\Psi, \Lambda)$ is the core of the procedure and is described in the next subsection.

\begin{definition}[$(\Psi,\Lambda)$-Critical Edge]
Given a failing edge $e$ and a swap edge $f = (v, u) \in S(e, v)$, and given two trees $\Psi, \Lambda$ of the centroid decomposition $\T$, a $(\Psi, \Lambda)$-\emph{critical edge} for $f$ is an edge $g =(x, y) \in \arg \max_{g' \in  S(e, V(\Psi) \cap U_e, V(\Lambda) \cap D_e)} \phi(f, g')$. When $\Psi=T$ we will refer to a $(\Psi,\Lambda)$-critical edge as a $\Lambda$-critical edge.
\end{definition}

Let $f=(v,u) \in S(e, v)$ and let $\Lambda$ be a tree of the centroid decomposition $\T$ such that $u \in V(\Lambda)$. Procedure~\texttt{\ref{alg:find_critical}} returns a $\Lambda$-critical edge for $f$, when edge $e$ fails (such an edge always exists as $f$ has one endpoint in $U_e$ and the other in $V(\Lambda) \cap D_e$). Notice that the call ~\texttt{\ref{alg:find_critical}}$(f, T)$ in Procedure~\texttt{\ref{alg:find_BCE}} computes a critical edge for $f$, since a $T$-critical edge for $f$ is actually a critical edge for $f$.

Procedure~\texttt{\ref{alg:find_critical}} uses as a subroutine Procedure~\texttt{\ref{alg:find_critical_top}}($f$, $\Psi$, $\Lambda$), which for the sake of clarity will be described in the next subsection. For the moment, it suffices to know that \texttt{\ref{alg:find_critical_top}} receives three inputs, i.e., edge $f=(v,u)$ and two subtrees $\Psi, \Lambda$ of the centroid decomposition $\T$ such that $v \in \Psi$ and, either $u \not \in V(\Lambda)$ or $\Lambda$ is the tree containing the sole vertex $u$, and it returns a $(\Psi, \Lambda)$-critical edge for $f$. If no such edge exists, then \texttt{\ref{alg:find_critical_top}} returns $\bot$ and we assume that $\phi(f, \bot)= -\infty$.

\newcommand{\lb}{(}
\newcommand{\rb}{)}
\begin{procedure}[t]{\caption{FindCritical($f=\lb v, u \rb, \Lambda$)}}
	\label{alg:find_critical}
	
	\footnotesize{	\lIf{$V(\Lambda) = \{ u \}$ }{
		\Return \texttt{FindCriticalCandidate}($f$, $T$, $\Lambda$)
	}

	$c \gets $ Centroid of $\Lambda$\;
	Let $j$ be the unique index in $\{0,1,2,3\}$ such that $u \in V(\tau_c^j)$\;

	\lIf{$c \in U_e$ } { \Return \FindCritical{$f, \tau_c^j$} }

	\BlankLine

		$\mathcal{G} \gets \{  \texttt{FindCriticalCandidate}(f,T,\tau_c^i): i=0,1,2,3 \wedge i \neq j \}$\tcp*{Here $c \in D_e$}
		$g_1 \gets  \arg \max_{g \in \mathcal{G}} \phi(f,g)$\;
		$g_2 \gets$ \FindCritical{$f, \tau_c^j$}\;
		\Return $\arg \max_{g \in \{ g_1, g_2 \}} \{ \phi(f, g) \}$\;				
}
\end{procedure}

\begin{lemma}
	Let $f=(v,u) \in S(e,v)$, and let $\Lambda$ be a tree of the centroid decomposition $\T$ such that $u \in V(\Lambda)$.
	Procedure~\texttt{\ref{alg:find_critical}}($f, \Lambda$) returns a $\Lambda$-critical edge for $f$.
\end{lemma}
\begin{proof}
	The proof is by induction on the cardinality of $V(\Lambda)$.

	If $|V(\Lambda)| = 1$, then the only vertex in $\Lambda$ must be $u$ and Procedure~\texttt{\ref{alg:find_critical}} invokes Procedure~\texttt{\ref{alg:find_critical_top}}$(f, T, \Lambda$). Hence, assuming such a procedure is correct, it returns a $(T, \Lambda)$-critical edge, i.e., a $\Lambda$-critical edge.
	If $|V(\Lambda)| > 1$ then we distinguish two cases, depending on the position of the centroid $c$ of $\Lambda$.

	If $c \in D_e$ it is sufficient to notice that a $\Lambda$-critical edge for $f$ must be incident to a tree $\tau_c^i$ for some $i=0,1,2,3$. Let $j$ be the unique index in $\{0,1,2,3\}$ such that $u \in V(\tau_c^j)$. If $j \neq i$ then, assuming Procedure~\texttt{\ref{alg:find_critical_top}} is correct, it returns a $(T, \Lambda)$-critical edge $g_1$ (and hence a $\Lambda$-critical edge) for $f$. Procedure~\texttt{\ref{alg:find_critical}} then returns either $g_1$ or another edge $g$ such that $\phi(f, g)=\phi(f, g_1)$. If $j=i$, the algorithm is recursively invoked and, since $|V(\tau_c^i)| < |V(\Lambda)|$ we know, by the induction hypothesis, that it correctly returns a $\tau_c^i$-critical edge for $f$, which is also $\Lambda$-critical edge for $f$.

	If $c \in U_e$, then 
	we know that there is at most one $\tau_c^i$ containing one or more vertices in $D_e$ (as otherwise the vertices in $V(\Lambda) \cap D_e$ would be disconnected in $\Lambda$, a contradiction). Moreover, since $u \in  V(\Lambda) \cap D_e$, there is exactly one such tree $\tau_c^i$, namely $\tau_c^j$. The algorithm recursively invokes itself on $\tau_c^j$ and, since $|V(\tau_c^j)| < |V(\Lambda)|$, we know, by induction hypothesis, that it returns a $\tau_c^j$-critical edge for $f$, which is also $\Lambda$-critical edge for $f$.
\end{proof}

\begin{lemma}\label{lemma:FindCritical}
	Procedure~\texttt{\ref{alg:find_critical}}$(f,\Lambda)$ requires $O(\Gamma_{\emph{\texttt{FCC}}} \cdot \log n)$ time, where $\Gamma_{\emph{\texttt{FCC}}}$ is the time required by an invocation of Procedure \texttt{\ref{alg:find_critical_top}}.
\end{lemma}
\begin{proof}
	Notice that Procedure~\texttt{\ref{alg:find_critical}} performs exactly one recursive invocation for each vertex of the tree $\T$ on the unique path between the root of $\T$ and $u$ in $\T$. The claim follows since the height of $\T$ is $O(\log n)$.
\end{proof}
In the next subsection, we show that $\Gamma_{\texttt{FCC}} =O(\log^2 n)$, and then we give our final result.

\subsection{Procedure FindCriticalCandidate}

In this subsection, we describe the core of the procedure that computes a critical edge for $f$. Let us first describe informally the main idea of this part. Let $b \in U_e$ and $c \in D_e$, and consider any two edges $f=(v,u),g=(x,y) \in S(e)$ such that $b$ (resp. $c$) is on the unique path from $x$ to $v$ (resp. from $y$ to $u$) in $T$ (see Figure \ref{fig:par-centr}). It turns out that the stretch factor of \emph{any} $f$ w.r.t. a \emph{given} $g$ can be though as a linear function $\Phi_{b,c,g}(t) = \alpha_{b,c}(g) \cdot t + \beta_{b,c}(g)$, where $\alpha_{b,c}(g)$ and $\beta_{b,c}(g)$ only depend on $g$. More precisely, we will have that $\phi(f,g) = \Phi_{b,c,g}(t_{b,c}(f))$, for a suitable value $t_{b,c}(f)$ which only depends on $f$. Hence, whenever we look for a critical edge for $f$, we can ask for a corresponding function $\Phi_{b,c,g}(t)$ with maximum value on $t_{b,c}(f)$. Since we do not know a priori the edge $f$ for which we need to compute a critical edge, we will maintain this information as the \emph{upper envelope} of a suitable set of functions. Let us make this idea more precise.

\begin{definition}[Upper Envelope]
	Let $\F=\{ \Phi_1, \Phi_2, \dots, \Phi_\ell \}$ be a finite set of functions, where $\Phi_i : \mathbb{R} \to \mathbb{R}$ for every $i=1,2,\dots,\ell$. The upper envelope of $\F$ is defined as $\displaystyle \UE_\F: t \in \mathbb{R} \mapsto  \arg \max_{\Phi \in \F}\Phi(t) \in 2^{\F}$.
\end{definition}

Let $b \in U_e$ and $c \in D_e$. Given an edge $f=(v, u)$, define $t_{b,c}(f)$ as the quantity $d_T(b, v)+ w(f) + d_T(u, c)$. Given an edge $g=(x, y)$, define $\alpha_{b,c}(g) = \frac{1}{w(g)}$ and $\beta_{b,c}(g) = \frac{d_T(x, b) + d_T(c, y)}{w(g)}$. Notice how, once $b$ and $c$ are fixed, $t_{b,c}(f)$ only depends on $f$ while $\alpha_{b,c}(g)$ and $\beta_{b,c}(g)$ only depend on $g$. Let $\Phi_{b,c,g}(t) = \alpha_{b,c}(g) \cdot t + \beta_{b,c}(g)$.

\begin{lemma}
	\label{lemma:linear_fnc}
	Let $f=(v,u) \in S(e,v)$.
	Let $b \in U_e$ and $c \in D_e$. Let $X$ (resp. $Y$) be a set of vertices $x \in U_e$ (resp. $y \in D_e$) such that vertex $b$ (resp. $c$) is on the unique path from $x$ to $v$ (resp. from $y$ to $u$) in $T$. 
	For every $g \in S(e, X, Y)$ we have $\phi(f,g) = \Phi_{b,c,g}(t_{b,c}(f))$.
\end{lemma}
\begin{proof}
Let $g = (x, y)$. We have:
\begin{align*}
	\phi(f, g) & = \frac{d_T(x, v) + w(f) + d_T(u, y)}{w(g)}
	= \frac{d_T(x, b) + d_T(b, v) + w(f) + d_T(u, c) + d_T(c, y)}{w(g)} \\ 
	& = \frac{d_T(b, v)+ w(f) + d_T(u, c)}{w(g)} + \frac{d_T(x, b) + d_T(c, y)}{w(g)}
	= \alpha_{b,c}(g) t_{b,c}(f) + \beta_{b,c}(g) \\
	& = \Phi_{b,c,g}(t_{b,c}(f)). \tag*{\qedhere}
\end{align*}
\end{proof}

\begin{figure}[t]
\floatbox[{\capbeside\thisfloatsetup{capbesideposition={right,top},capbesidewidth=5.5cm}}]{figure}[\FBwidth]
{\caption{Illustration of Lemma \ref{lemma:critical_edge_query}. $f$ is a swap edge for $e$, $\Psi$ and $\Lambda$ are two trees of the centroid decomposition, and $b$ and $c$ are their corresponding parent centroids. $g$ is a potential $(\Psi, \Lambda)$-critical edge for $f$. Notice that the unique path from $x$ to $v$ (resp. from $y$ to $u$) passes through $b$ (resp. $c$).}\label{fig:par-centr}}
{\includegraphics[scale=0.8]{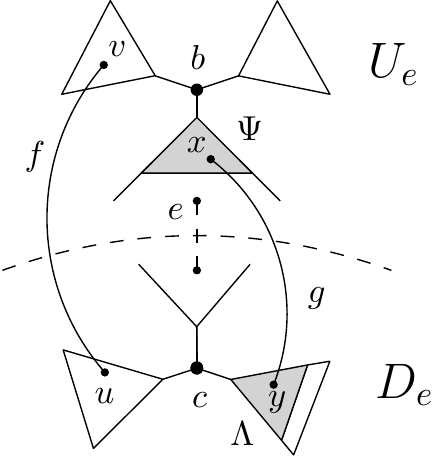}\hspace{2cm}}
\end{figure}

\begin{definition}[Parent centroid]
	Let $\tau$ be a tree of the centroid decomposition $\T$. The \emph{parent centroid} of $\tau$ is the centroid of the parent of $\tau$ in $\T$.
\end{definition}

Lemma~\ref{lemma:linear_fnc} is instrumental to proving the following (see Figure \ref{fig:par-centr}):
\begin{lemma}
	\label{lemma:critical_edge_query}
	Let $f=(v,u) \in S(e,v)$, and let $\Psi, \Lambda$ be two trees of the centroid decomposition of $T$ such that the following conditions hold: (i) $v \not\in V(\Psi)$ or $V(\Psi) = \{ v \}$, and (ii) $u \not\in V(\Lambda)$ or $V(\Lambda) = \{ u \}$. Let $b$ (resp. $c$) be the parent centroid of $\Psi$ (resp. $\Lambda$), and assume that $b \in U_e$ (resp. $c \in D_e$).
	Then, an edge $g$ is a $(\Psi, \Lambda)$-critical edge for $f$ if and only if $\Phi_{b,c,g} \in \UE_{\F}(t_{b,c}(f))$ where $\F=\{ \Phi_{b,c,g'} : g' \in S(e, V(\Psi) \cap U_e, V(\Lambda) \cap D_e) \}$.
\end{lemma}
\begin{proof}
	First of all we show the following property of the centroid decomposition $\T$: let $p,q \in V(T)$, and suppose that the unique path in $\T$ between the leaf nodes associated with $p$ and $q$ contains a node whose corresponding centroid is $z$. Then, the unique path between $p$ and $q$ in $T$ contains $z$.  Indeed, if $z$ is either $p$ or $q$, the property is trivially true. On the other hand, suppose that $z \not\in \{p, q\}$, and let $\tau$ be the subtree of $T$ associated with $z$ in $\T$. Then, let $\tau_z^i$ be the child subtree of $\tau$ containing $p$. Observe that $q$ is not in $\tau_z^i$. Moreover, by construction, each path from a node of $\tau_z^i$, and in particular from $p$, to any node outside $\tau_z^i$, and in particular to $q$, must pass through $z$.
	
	We now prove the claim. If $V(\Psi) = \{ v \}$ (resp. $V(\Lambda)=\{u\}$) then it follows from Lemma~\ref{lemma:linear_fnc} by choosing  $X = \{ v \}$ and $Y = V(\Lambda) \cap D_e$ (resp. $X = V(\Psi) \cap U_e$ and $Y = \{ u \}$).
	The complementary case is the one in which $v \not\in V(\Psi)$ and $u \not\in V(\Lambda)$.
	Consider the vertices $v$ and $b$ (resp. $u$ and $c$) in $\T$ and notice that $v$ (resp. $u$) cannot be an ancestor of $b$ (resp. $c$). Indeed, if that were the case, then the subtree of $T$ induced by the vertices in $V(\Psi)$ (resp. $V(\Lambda)$) would contain $b$ (resp. $c$) contradicting the hypothesis.
	Hence, the path from any vertex in $V(\Psi)$ to $v$ (resp. $V(\Lambda)$ to $u$) traverses $b$ (resp. $c$) in $\T$  and therefore the same holds in $T$.
	The claim follows by invoking Lemma~\ref{lemma:linear_fnc} with $X=V(\Psi) \cap U_e$ and $Y = V(\Lambda) \cap D_e$.
\end{proof}

Lemma~\ref{lemma:critical_edge_query} allows us to design a recursive procedure to compute a $(\Psi, \Lambda)$-critical edge for $f$ (see Procedure~\texttt{\ref{alg:find_critical_top}}). To this aim we will make use of a data structure $\U_e$ that, for each edge $f \in S(e)$, and for each pair of trees $\Psi, \Lambda$ of the centroid decomposition, can perform a query operation that we name $\U_e(f, \Psi, \Lambda)$. This query reports an edge whose function $\Phi_{b,c,g}$ is in $\UE_{\F}(t_{b,c}(f))$ where $b$ and $c$ are the parent centroids of $\Psi$ and $\Lambda$, respectively, and $\F=\{ \Phi_{b,c,g'} : g' \in S(e, V(\Psi) \cap U_e, V(\Lambda) \cap D_e) \}$.

Next two lemmas show the correctness and the running time of the procedure:

\begin{lemma}
\label{lemma:critical_edge_algorithm}
Let be given an edge $f=(v,u) \in S(e, v)$ and two trees $\Psi, \Lambda$ of the centroid decomposition such that: (i) $v \in V(\Psi)$, and (ii) $u \not\in V(\Lambda)$ or $V(\Lambda) = \{ u \}$. Then, Procedure~\texttt{\ref{alg:find_critical_top}}($f, \Psi, \Lambda $) computes a $(\Psi, \Lambda)$-critical edge for $f$.
\end{lemma}
\begin{proof}
	First of all notice that if $V(\Lambda) \cap D_e = \emptyset$, then the algorithm correctly returns $\bot$. We now prove the claim by induction on $|V(\Psi)|$. If $|V(\Psi)| = 1$, then the only vertex in $\Psi$ must be $v$ and Procedure~\texttt{\ref{alg:find_critical_top}} queries $\U_e$ for $\U_e(f, \Psi, \Lambda)$. By Lemma~\ref{lemma:critical_edge_query}, the returned edge is a $(\Psi, \Lambda)$-critical edge for $f$. 
	If $|V(\Psi)| > 1$ then we distinguish two cases, depending on the position of the centroid $b$ of $\Psi$. If $b \in U_e$ it is sufficient to notice that a $(\Psi, \Lambda)$-critical edge for $f$ must be incident to a tree $\tau_b^i$ for some $i=0,1,2,3$. Let $j$ be the unique index in $\{0,1,2,3\}$ such that $v \in V(\tau_b^j)$. If $j \neq i$ then, by Lemma~\ref{lemma:critical_edge_query}, the query $\U_e(f, \tau_b^i, \Lambda)$ returns a $(\tau_b^i, \Lambda)$-critical edge $g'$ (and hence $g'$ is also a $(\Psi, \Lambda)$-critical edge) for $f$. Procedure~\texttt{\ref{alg:find_critical}} then returns either $g'$ or another edge $g$ such that $\phi(f, g)=\phi(f, g')$. If $j=i$, the algorithm is recursively invoked and, since $|V(|\tau_b^i|)| < |V(\Psi)|$ we know, by the induction hypothesis, that it returns a $(\tau_b^i, \Lambda)$-critical edge for $f$, which is also $(\Psi, \Lambda)$-critical edge for $f$.
If $b \in D_e$, then there is at most one $\tau_b^i$ containing at least one vertex in $U_e$ (as the converse would imply that the vertices in $V(\Psi) \cap U_e$ are disconnected in $\Psi$, a contradiction). Moreover, since $v \in V(\Psi) \cap U_e$, there is exactly one such tree $\tau_b^i$, namely $\tau_b^j$. The algorithm recursively invokes itself on $\tau_b^j$ and we know, by induction hypothesis, that it returns a $\tau_b^j$-critical edge for $f$, which is also $(\Psi,\Lambda)$-critical edge for $f$.
\end{proof}

\begin{procedure}[t]{\caption{FindCriticalCandidate($f=\lb v, u \rb, \Psi, \Lambda $)}}

	\label{alg:find_critical_top}
	\footnotesize{\lIf{$V(\Lambda) \cap D_e = \emptyset$}{\Return $\bot$}
	\lIf{$V(\Psi) = \{ v \}$}{\Return $\U_e(f, \Psi, \Lambda)$}	

	$b \gets $ Centroid of $\Psi$\;	
	Let $j$ be the unique index in $\{0,1,2,3\}$ such that $v \in V(\tau_b^j)$\;

	\lIf{$b \in D_e$}
	{\Return \FindCriticalCandidate{$f, \tau_b^j, \Lambda$}}
	
	$\mathcal{G} \gets \{  \U_e(f, \tau_b^i, \Lambda) : i=0,1,2,3 \wedge i \neq j \}$\tcp*{Here $b \in U_e$}
	$g_1 \gets  \arg \max_{g \in \mathcal{G}} \phi(f,g)$\;
	$g_2 \gets $ \FindCriticalCandidate{$f, \tau_b^j, \Lambda$}\;
	\Return $\arg \max_{g \in \{ g_1, g_2 \}} \{ \phi(f, g) \}$\;		
}
\end{procedure}

\begin{lemma}\label{lemma:FindCriticalCandidate}
 Procedure~\texttt{\ref{alg:find_critical_top}}($f, \Psi, \Lambda $) requires $O(\Gamma_{\U_e} \cdot \log n)$ time, where $\Gamma_{\U_e}$ is the time required by a query on $\U_e$.
\end{lemma}
\begin{proof}
	Notice that Procedure~\texttt{\ref{alg:find_critical_top}} performs exactly one recursive invocation for each vertex of the tree $\T$ on the unique path between the root of $\T$ and $u$ in $\T$. The claim follows since the height of $\T$ is $O(\log n)$.
\end{proof}

Thus, to get the promised running time of $O(\log^2n)$ for $\Gamma_{\texttt{FCC}}$, we are left to prove that $\Gamma_{\U_e}=O(\log n)$.
Actually, such a bound can be obtained by suitably implementing $\U_e$ in such a way that all the underlying upper envelope functions are efficiently maintained. Due to space limitation, this technical part is postponed to the appendix.
By combining all the lemmas, and by observing that we need $O(n^2)$ space to handle the top-trees, and $O(m \log ^2 n)$ space to implement each $\U_e$ (as shown in the appendix), we eventually can give the following:

\begin{theorem}\label{main}
	The \absets problem can be solved in $O(n^2 \log^4 n)$ time and $O(n^2+m \log^2 n)$ space.
\end{theorem}

\bibliographystyle{plainurl}
\bibliography{bibliography}

\newpage
\appendix
\section{Reducing the degree of \texorpdfstring{$T$}{T}}\label{sec:degree_reduction}
Here we show a linear time reduction that transforms $G$ and a rooted spanning tree $T$ of $G$ into an equivalent graph $G'=(V(G'),E(G'),\overline {w})$ and a corresponding binary spanning tree $T'$ such that $|V(G')|=\Theta(n)$ and $|E(G')|=\Theta(m)$.

Initially, $G', w'$, and $T'$ coincide with $G,w$, and $T$, respectively. We iteratively search for a vertex $u$ in $T'$ that has 3 or more children, and we lower its degree. Let $v_1,\dots,v_h$, with $h\geq 3$, be the children of $u$. We remove all the edges $\{(u,v_i):1\leq i\leq h\}$ from both $G'$ and $T'$, then we add to both $G'$ and $T'$ a binary tree whose root coincides with $u$, and that has exactly $h$ leaves $x_1,\dots,x_h$. We assign weight $w'(e)=0$ to all the edges $e$ of this tree. Finally, we add to $G'$ and $T'$ an edge $(x_i,v_i)$ for each $1\leq i \leq h$, and we set $w'(x_i,v_i)=w(u,v_i)$. An example of such a transformation is shown in Figure~\ref{fig:degree-reduction}.

\begin{figure}[ht!]
    \centering
    \includegraphics[scale=1.2]{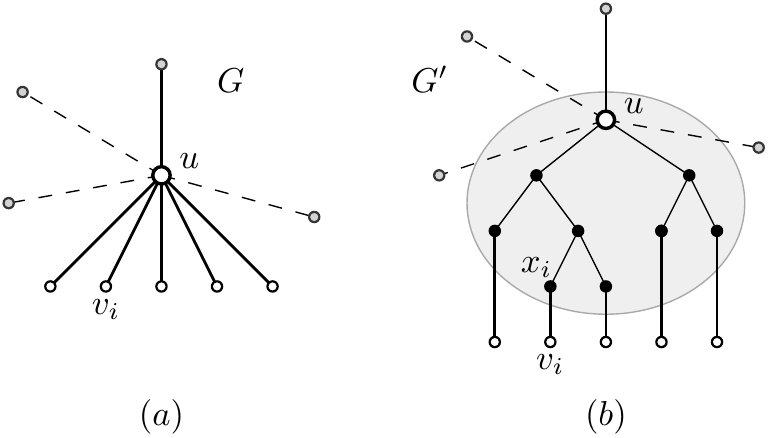}
    \caption{Reducing the degree of vertices in $G$: on the left side, the tree $T$ (solid edges) embedded in $G$, on the right side the superimposition of the binary tree to $T$ in order to get a maximum degree of $3$. Solid edges in the gray area have weight $0$, while the weight of $(x_i ,v_i )$ is $w(u,v_i )$.}
    \label{fig:degree-reduction}
\end{figure}

Clearly, $|V(G')|=O(|V(G)|)$, $|E(G')|=O(|E(G)|)$, and, moreover, the computation of $G'$ and $T'$ requires linear time. Now, observe that, for every $a,b \in V(G)$, it holds that $d_{T_{e/f}}(a,b)=d_{T'_{e/f}}(a,b)$. Furthermore, for every edge $e=(u,v_i)$ of $T$, $f$ is a swap edge for $e$ in $T$ iff $f$ is a swap edge for the edge $(z,v_i)$ in $T'$, where $z$ is the parent of $v_i$ in $T'$. As a consequence, we can conclude that, for every edge $e=(u,v_i)$ of $T$, $f \in S(e)$ is a BSE for $e$ w.r.t. $T$ iff $f$ is a BSE for the edge $(z,v_i)$ w.r.t. $T'$, where $z$ is the parent of $v_i$ in $T'$.

\section{Selecting edge \texorpdfstring{$f$}{f} in Step \ref{ln:find_f} of Procedure~\ref{alg:find_BCE}(\texorpdfstring{$\Lambda$}{Lambda})}
\label{sec:step 7}
In this section we show how to efficiently find an edge $f=(v,u) \in S(e)$ minimizing $w(f) + d_T(u, c)$, where $c$ is a vertex in $D_e$, as Procedure~\texttt{\ref{alg:find_BCE}} requires. We will show how this problem can be solved by using \emph{top-trees} \cite{AHLT05}.

A top-tree is a dynamic data structure that maintains a (weighted) forest $F$ of trees under \emph{link}  (i.e., edge-insertion) and \emph{cut} (i.e., edge-deletion) operations. Moreover, some of the vertices of $F$ can be \emph{marked} and the top-tree is able to perform \emph{closest marked vertex} (CMV, for short) queries, i.e., it can report the marked vertex that is closest to a given vertex $z$. A top-tree on $n$ vertices can be built in linear time and each of all the aforementioned operations requires $O(\log n)$ time.

We maintain a top-tree $\Upsilon_v$ of size $O(n)$ for each vertex $v \in V(T)$, and so we use a total of $O(n^2)$ space. Each of these top-trees is the tree $T$ augmented with some additional marked vertices. More precisely, for each $v \in V(T)$ and for each edge $f=(v,u) \in E(G) \setminus E(T)$ we add to $\Upsilon_v$ a marked vertex $\overline{u}$ and the edge $(u,\overline{u})$ with a weight of $w(f)$ (see Figure~\ref{fig:top-tree}).

Whenever we are finding a $v$-BSE for $e$ and we need to find the edge $f$ minimizing $w(f) + d_T(u, c)$ we do the following: (i) we cut the edge $e$ from $\Upsilon_v$, (ii) we perform a CMV query on $\Upsilon_v$ to find the closest marked vertex $\overline{u}$ to $c$, if any, (iii) we undo the cut operation by linking the endpoints of $e$ in $\Upsilon_v$, and finally (iv) we return the edge $(v,u)$ (or $\perp$ if no $\overline{u}$ has been found).

\begin{figure}
	\centering
	\includegraphics[scale=0.8]{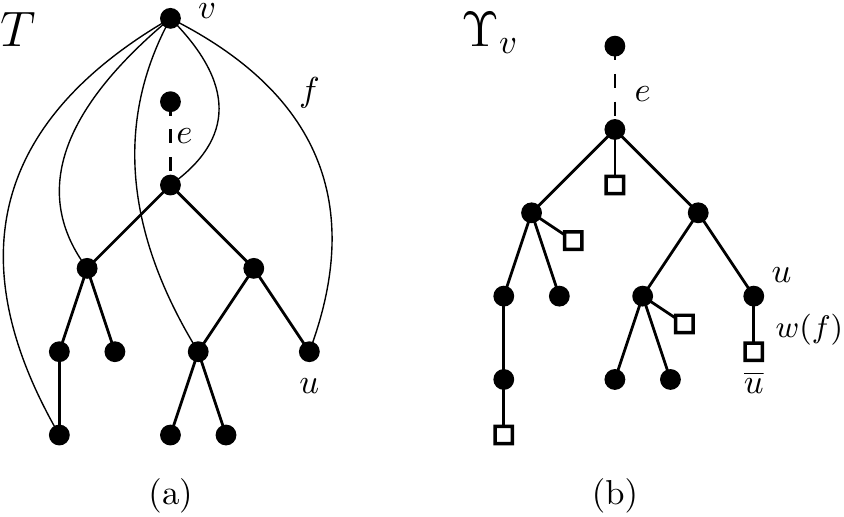}
	\caption{(a) The subtree of $T$ induced by $D_e$ along with swap edges in $S(e,v)$, and (b) the corresponding top-tree $\Upsilon_v$. The black vertices of $\Upsilon_v$ are the same of the tree $T$. For each $u \in V(T)$ such that $f=(v,u) \in S(e,v)$, $\Upsilon_v$ contains an additional vertex $\overline{u}$ (shown as a white square), and a corresponding edge $(u,\overline{u})$ with weight $w(f)$.}
	\label{fig:top-tree}
\end{figure}

\begin{lemma}
Let $e \in E(T)$ be a failing edge and let $c \in D_e$. An edge $f=(v,u) \in S(e,v)$ minimizing $w(f) + d_T(u, c)$ can be found in $O(\log n)$ time. Moreover, all the top-trees  $\Upsilon_v$ can be initialized in $O(n^2)$ time.
\end{lemma}
\begin{proof}
	Each of the $n$ top-trees $\Upsilon_v$ can be built in time $O(n)$ by explicitly considering all the edges in $E(v)$ (notice that  $\Upsilon_v$ contains at most $2n$ vertices  as there can be at most one marked vertex per vertex in $V(T)$).

	As for the time complexity of finding edge $f$, it immediately follows from the fact that we perform a constant number of link, cut, and CMV query operations, hence we only need to argue about correctness.
	
	Notice that after we cut edge $e$ from $\Upsilon_v$ in step (i), the tree $T^\prime$ of $\Upsilon_v$ containing $c$ has exactly one (distinct) marked vertex $\overline{u}$ for each edge $(u,v) \in S(e,v)$. The claim follows as, by construction, the distance from $c$ to $\overline{u}$ in $T^\prime$ is $d_{T^\prime}(c, \overline{u}) = d_{T^\prime}(c, u) + d_{T^\prime}(u, \overline{u}) = d_T(c, u) + w(f)$.
\end{proof}

\section{Dynamic maintenance of the upper envelopes}
\label{sec:UE}
Procedure~\texttt{\ref{alg:find_critical_top}} needs the auxiliary structure $\U_e$. Explicitly building such a structure for each edge $e$ would be too expensive. In the remaining of this section we show how all the $\U_e$'s can be built in $O(m \log^4 n)$ time and $O(m \log^2 n)$ space. The idea is to exploit the order in which failing edges are considered, so as to reuse previously computed information to build $\U_e$.

We implement $\U_e$ as a dictionary that allows us to add, delete, and search for elements in $O(\log n)$ time per operation. Each element of $\U_e$ is a  data structure that can store a set $\F$ of linear functions and is able (i) to dynamically add/remove a function to/from $\F$ in $O(\log |\F|)$ time, (ii) given $t \in \mathbb{R}$, to report a function in $\UE_\F(t)$ in $O(\log  |\F|)$ amortized time \cite{brodal2002dynamic}.

Each data structure in the dictionary is associated with a pair  $\Psi, \Lambda$ of trees of $\T$ and will contain all the functions $\Phi_{b,c,g}$ where $b$ and $c$ are the parent centroids of $\Psi$ and $\Lambda$, respectively, and $g \in S(e, V(\Psi) \cap U_e, V(\Lambda) \cap D_e)$. We name such a structure $\H^e_{(\Psi,\Lambda)}$. The pair $(\Psi, \Lambda)$ is also the key of $\H^e_{(\Psi ,\Lambda)}$ in the dictionary.

Observe now that we can answer the query $\U_e (f, \Psi, \Lambda )$ (used in \texttt{\ref{alg:find_critical_top}}) in $O(\log n)$ amortized time, as follows: we search for $\H^e_{\Psi , \Lambda} \in \U_e$ in $O(\log n)$ time, and then we perform a query operation on $\H^e_{\Psi , \Lambda}$ with $t = t_{b,c}(f)$ where $b$ and $c$ are the parent centroids of $\Psi$ and $\Lambda$, respectively (see Lemma \ref{lemma:linear_fnc}).

We now show how to build and maintain the $\U_e$'s.
Remember that we process the edges $e \in E(T)$ in a bottom-up fashion. Let $T_e$ be the subtree of $T$ induced by $D_e$. Whenever $T_e$ consists of a single vertex, we build $\U_e$ from scratch. If $T_e$ contains 2 or more vertices then there are at most two edges $e_1$, $e_2 \in E(T_e)$ that are incident to $e$. We build $\U_e$ by merging $\U_{e_1}$ and $\U_{e_2}$. This merge operation consists of a \emph{join} step followed by an \emph{update} step.

Whenever we add a function $\Phi_{b,c,g}$ to a structure $\H^e_{(\Psi,\Lambda)}$ of $\U_e$ and we are either performing the update step or we are building $\U_e$ from scratch, we say that we \emph{insert} $\Phi_{b,c,g}$ into $\U_e$.
We associate a non-negative integer $\nu_e$ to $\U_e$ that we call \emph{virtual size} of $\U_e$. The virtual size of $\U_e$ is the overall number of inserts that have been performed either on $\U_e$ itself or on any other $\H^{e'}_{(\Psi, \Lambda)}$ such that $e^\prime$ is an edge of $T_e$.

\subsection{Building \texorpdfstring{$\U_e$}{Ue} from scratch}

We start by creating an empty dictionary $\U_e$ (initially $\nu_e=0$).
Since we are building $\U_e$ from scratch, $T_e$ contains only one vertex, say $y$. For each edge $g =(x,y) \in S(e,U_e,y)$, we explicitly consider all the pairs of trees $(\Psi , \Lambda)$, such that $\Psi$ contains $x$ and $\Lambda$ contains $y$, and we let $b$ and $c$ be the parent centroids of $\Psi$ and $\Lambda$, respectively.
We look for $\H^e_{(\Psi,\Lambda)}$ in the dictionary of $\U_e$, if $\H^e_{(\Psi,\Lambda)}$ already exists, we add $\Phi_{b,c,g}$ to $\H^e_{(\Psi,\Lambda)}$. If $\H^e_{(\Psi,\Lambda)}$ is not found, we create a new empty structure $\H^e_{(\Psi,\Lambda)}$, we add $\Phi_{b,c,g}$ into $\H^e_{(\Psi,\Lambda)}$, and we add $\H^e_{(\Psi,\Lambda)}$ to $\U_e$. In both cases we have that $\Phi_{b,c,g}$ is \emph{inserted} into $\U_e$ and hence we increase $\nu_e$ by $1$.

\subsection{Building \texorpdfstring{$\U_e$}{Ue} by merging}

Let $e=(p,q)$ and remember that $T_e$ contains more than $1$ vertex. Since $q$ has degree at most $3$ in $T$, there are either $1$ or $2$ edges in $T_e$ that are incident to $q$. Here we will discuss the case in which those edges are exactly $2$ (as the case in which $q$ is incident to only one edge is simpler).

Let $e_1$, $e_2$ be the two edges incident to $q$ in $T_e$. We will merge $\U_{e_1}$ and $\U_{e_2}$ in order to obtain $\U_e$. This operation is destructive, i.e., $\U_{e_1}$ and $\U_{e_2}$ will no longer exist at the end of the merge operation. Notice, however, that since we are processing the edges of $T$ in a bottom-up fashion, $\U_{e_1}$ and $\U_{e_2}$ will no longer be needed by the algorithm.

The merge operation consists of two steps: the \emph{join step} and the \emph{update step}.

\subsubsection{The join step}

W.l.o.g., let $\nu_{e_1} \ge \nu_{e_2}$.
We start by renaming $\U_{e_1}$ to $\U_e$ (so that all the structures that belong to the dictionary of $\U_{e_1}$ that were named $\H^{e_1}_{(\Psi,\Lambda)}$ are now named $\H^e_{(\Psi,\Lambda)}$).

Now, for each structure $\H^{e_2}_{(\Psi,\Lambda)}$ in $\U_{e_2}$, we first search for the structure $\H^{e}_{(\Psi,\Lambda)}$ in $\U_{e}$ and, if such a structure is not found, we add new empty structure $\H^e_{(\Psi,\Lambda)}$ to $\U_e$.
Then, we \emph{move} each function $\Phi_{b,c,g}$ in $\H^{e_2}_{(\Psi,\Lambda)}$ to $\H^{e}_{(\Psi,\Lambda)}$, i.e., we remove $\Phi_{b,c,g}$ from $\H^{e_2}_{(\Psi,\Lambda)}$ and, if $\Phi_{b,c,g}$ is not in $\H^{e}_{(\Psi,\Lambda)}$, we add it to $\H^{e}_{(\Psi,\Lambda)}$.
Finally, after all the structures $\H^{e_2}_{(\Psi,\Lambda)}$ in $\U_{e_2}$ have been considered, we destroy $\U_{e_2}$ and we set $\nu_e$ to $\nu_{e_1} + \nu_{e_2}$.

\subsubsection{The update step}

After the merge step is completed, $\U_e$ contains all the functions corresponding to the edges $g$ in $S(e_1) \cup S(e_2)$.

Notice, however, that all the edges $(x,y)$ such that the \emph{lowest common ancestor} (LCA) of $x$ and $y$ in $T$ is $q$ are both in $S(e_1)$ and $S(e_2)$ but they do not belong to $S(e)$, and hence they should not appear in $\U_e$. On the converse, the edges in $S(e, U_e, q)$ are neither in $S(e_1)$ nor in $S(e_2)$ but they belong to $S(e)$, hence their corresponding functions should be added to $\U_e$. This is exactly the goal of the update step.

We start by deleting the extra functions from $\U_e$.
We iterate over each edge $g=(x,y)$ such that the LCA of $x$ and $y$ is $q$ and, for each pair of trees $(\Psi , \Lambda)$ such that $\Psi$ contains $x$ and $\Lambda$ contains $y$, we delete $\Phi_{b,c,g}$ from $\H^e_{( \Psi, \Lambda )}$ where $b$ and $c$ are the parent centroids of $\Psi$ and $\Lambda$ respectively. If $\H^e_{( \Psi, \Lambda )}$ becomes empty, we also delete $\H^e_{( \Psi, \Lambda )}$ from $\U_e$.

We now add the missing functions to $\U_e$.
For each $g=(x,q) \in S(e, U_e, q )$, and for each pair of trees $(\Psi, \Lambda)$, such that $\Psi$ contains $x$ and $\Lambda$ contains $q$, we first search for $\H^e_{( \Psi, \Lambda )}$ in $\U_e$ and, if it does not exist, we add new empty structure $\H^e_{( \Psi, \Lambda )}$ to $\U_e$. Then, we add $\Phi_{b,c,g}$ to $\H^e_{( \Psi, \Lambda )}$, where $b$ and $c$ are the parent centroids of $\Psi$ and $\Lambda$ respectively. We increase $\nu_e$ by $1$ to account for this insertion.

\subsection{Analysis}\label{sec:analysis UE}
Here we bound the time required to dynamically maintain all the upper envelope structures.
\begin{lemma}
	\label{lemma:number_of_functions_Phi}
	The overall number of distinct functions $\Phi_{b,c,g}$ ever inserted into at least one of the structures $\U_e$ is $O(m \log^2 n)$.
\end{lemma}
\begin{proof}
	Let us consider any edge $g=(x,y) \in E(G) \setminus E(T)$.	If a function $\Phi_{b,c,g}$ associated with $g$ is inserted into any $\U_e$, this means that it added to some $\H^e_{(\Psi, \Lambda )}$ such that $x \in V(\Psi)$, $y \in V(\Lambda)$, and $b$ (resp. $c$) is the parent centroids of $\Psi$ (resp. $\Lambda$). Notice that there are $O(\log n)$ trees $\tau$ of the centroid decomposition $\T$ that contain $x$ (resp. $y$), meaning that there are $O(\log^2 n)$ functions $\Phi_{b,c,g}$ associated to $g$. The claim follows by summing over all the edges in $E(G) \setminus E(T)$.
\end{proof}

\begin{lemma}
	\label{lemma:history_insertions}
	Each function $\Phi_{b,c,g}$ contributes at most $2$ to the virtual size $\nu_e$ of any $\U_e$.
\end{lemma}
\begin{proof}
	It suffices to bound the overall number of insertions of $\Phi_{b,c,g}$ (regardless of the structure $\U_e$ into which $\Phi_{b,c,g}$ is inserted).
	To this aim, consider the edge $g=(x,y)$ associated with $\Phi_{b,c,g}$ and, w.l.o.g., let $x$ be the vertex that is closest to the root of $T$.  Let also $e_x$ (resp. $e_y)$ be the edge from the parent of $x$ to $x$ (resp. from the parent of $y$ to $y$) in $T$.
	We distinguish two cases depending on the relative positions of $x$ and $y$ in $T$.
	If $x$ is an ancestor of $y$, then  $\Phi_{b,c,g}$ is only inserted in $\U_{e_y}$.
	Indeed, $g$ belongs to $S(e_y)$ but it does not belong to any $S(e')$ where $e' \in E(T_{e_y})$ is incident to $e_y$ in $T_{e_y}$. For any other pair of edges $e'',e''' \in E(T)$ such that $e'''$ is incident to $e''$ in $T_{e''}$ we have that either $g$ does not belong to $S(e'')$, or it belongs to both $S(e'')$ and $S(e''')$, and hence $\Phi_{b,c,g}$ is not added to $\U_{e''}$.
	If $x$ is not an ancestor of $y$, then a similar argument shows that $\Phi_{b,c,g}$ can only be inserted in $\U_{e_x}$ and in $\U_{e_y}$. The claim follows.
\end{proof}

\begin{lemma}
	\label{lemma:phi_moves}
	Each function $\Phi_{b,c,g}$ is moved $O(\log n)$ times.
\end{lemma}
\begin{proof}
	When a function is moved from any $\U_{e_2}$ to $\U_{e}$ it is because we are merging
	$\U_{e_1}$ with $\U_{e_2}$, where $e_1$ and $e_2$ are edges incident to $e$ in $T_e$.
	Notice that, before the merge operation takes place, we must have $\nu_{e_1} \ge \nu_{e_2}$ and hence, at the end of the merge operation, $\nu_e \ge \nu_{e_1} + \nu_{e_2} \ge 2 \nu_{e_2}$.
	In other words, each time a function $\Phi_{b,c,g}$ is moved, we have that the \emph{virtual size} of the structure to which $\Phi_{b,c,g}$ belongs at least doubles. Therefore, after a function has been moved $r$ times, the structure containing $\Phi_{b,c,g}$ must have a virtual size of at least $2^r$.
	
	Notice now that Lemma~\ref{lemma:number_of_functions_Phi} and Lemma~\ref{lemma:history_insertions} imply an upper bound of $O(m \log^2 n)$ to the virtual size of any $\U_e$. We can conclude that a function can be moved $O( \log (m \log^2 n) ) = O( \log n )$ times.
\end{proof}

\begin{proposition}
	The total time spent building and merging all the data structures $\U_e$ is $O(m \log^4 n)$.
\end{proposition}
\begin{proof}
	From Lemma~\ref{lemma:number_of_functions_Phi} and Lemma~\ref{lemma:history_insertions} we have that the total number of insertions of functions $\Phi_{b,c,g}$ into the structures $\H^e_{(\Psi, \Lambda )}$ is $O(m \log^2 n)$ and, since each insertion requires time $O(\log n)$, the total time spent due to insertions is $O(m \log^3 n)$. Moreover, since each function is deleted at most once, and a deletion takes $O(\log n)$ time, we have
that the total time spent for deleting functions is $O(m \log^3 n)$.

	Concerning moving of functions, by Lemma~\ref{lemma:phi_moves} we have that every function is moved $O(\log n)$ times. Since there are $O(m \log^2 n)$ functions, as shown by Lemma~\ref{lemma:number_of_functions_Phi}, and a function can be moved in $O(\log n)$ time, we have that the total time spent moving functions is $O(m \log^4 n)$.
\end{proof}

\end{document}